\documentclass[journal]{IEEEtran}

\usepackage{url}

\usepackage{cite}
\usepackage{amsmath,amssymb,amsfonts}
\usepackage{textcomp}
\usepackage{xcolor}
\usepackage{tikz, graphics, color, float, epsf}
\usepackage{amssymb, amsmath, mathtools, amsthm}
\usepackage[hidelinks]{hyperref}
\usepackage[ruled,vlined]{algorithm2e}
\def\BibTeX{{\rm B\kern-.05em{\sc i\kern-.025em b}\kern-.08em
    T\kern-.1667em\lower.7ex\hbox{E}\kern-.125emX}}

\DeclareMathOperator*{\argmin}{arg\,min}

\newtheorem{theorem}{Theorem}
\newtheorem{lemma}[theorem]{Lemma}
\newtheorem{proposition}[theorem]{Proposition}

\theoremstyle{remark}
\newtheorem{remark}{Remark}

\SetKw{Assert}{Assert}
\SetKwInOut{Require}{Require}

\usepackage{graphicx}
\usepackage{balance}
\usepackage{eso-pic}

\begin{document}

\title{Fast Graph Laplacian Estimation using\\Effective Resistance}

\author{Christoffer Kjellson, Claudio Altafini, and Emma Tegling
\thanks{C. Kjellson and E. Tegling are with the Department of Automatic Control, Lund University. Email: \{{\tt\small{christoffer.kjellson, emma.tegling}\}@control.lth.se.}  C. Altafini is with the Division of Automatic Control, Department of Electrical Engineering, Linköping University. E-mail: {\tt\small{claudio.altafini@liu.se}}. All are with the ELLIIT Strategic Research Area. }
\thanks{This work is partially funded by the Wallenberg AI, Autonomous Systems and Software Program (WASP) funded by the Knut and Alice Wallenberg Foundation and by the Swedish Research Council (grant 2024-04772 to C.A.).}}

\maketitle

\AddToShipoutPictureFG*{%
  \AtPageLowerLeft{%
    \put(40,20){%
      \parbox{\dimexpr\paperwidth-80pt\relax}{%
        \centering\footnotesize
        © 2026 IEEE.  Personal use of this material is permitted.  Permission from IEEE must be obtained for all other uses, in any current or future media, including reprinting/republishing this material for advertising or promotional purposes, creating new collective works, for resale or redistribution to servers or lists, or reuse of any copyrighted component of this work in other works.
      }%
    }%
  }%
}

\begin{abstract}
Inferring network topology from noisy node observations is a central problem in graph signal processing. In this paper, we consider Laplacian-constrained graph estimation for Gaussian Markov random fields, focusing on the underdetermined regime in which the number of samples is smaller than the number of graph nodes. Existing approaches often formulate the problem as a sparsity-regularized maximum-likelihood estimation problem. While effective, such methods typically require iterative optimization and are often computationally demanding, particularly under Laplacian constraints. Instead, we propose a non-iterative estimator of graph Laplacians that uses effective resistance for regularization, and evaluate the method using a simple sparsification procedure. Experiments show that with some trade-off in edge and weight recovery on the considered dataset, computational cost for moderately sized graphs can be substantially reduced.
\end{abstract}


\IEEEpeerreviewmaketitle

\section{Introduction}

\IEEEPARstart{I}{n} graph signal processing, observed signals are often governed by an underlying graph that is not directly known. Recovering this graph from noisy node observations is therefore a central problem and has applications in physical, biological, and social systems. Useful overviews of the topic are found in~\cite{ortega_signalprocessing} and~\cite{Mateos_2019}.

A classical approach for the case in which the signals on the graph are jointly Gaussian and the sought-after graph is assumed to be sparse is the graphical lasso~\cite{friedman},~\cite{yuan_lin},~\cite{banerjee}. This method originates from work on covariance/model selection~\cite{dempster, cox1996multivariate}, and is based on maximum likelihood estimation (MLE) regularized by sparsity. Such methods are appropriate for general precision matrices, but do not enforce the structure of a graph Laplacian, which is relevant in a wide variety of signal processing and machine learning applications~\cite{egilmez2017}. A separate framework estimates a Laplacian under the assumption of smooth signals on the graph, as described, for example, in~\cite{dong} and~\cite{kalofolias}, without using maximum likelihood estimation.

In this paper, we consider the special case where the covariance matrix equals the pseudoinverse of the graph Laplacian, i.e., $\Sigma=L^\dagger$, on a Gaussian Markov random field. This is often referred to as an intrinsic Gaussian Markov random field, or equivalently, a degenerate Gaussian model on the subspace orthogonal to $\mathbf{1}$. The goal is then to reconstruct the $p\times p$ Laplacian from $n<p$ independent and identically distributed samples from $\mathcal{N}(0,L^\dagger)$.

Laplacian-constrained MLE was described and implemented in~\cite{egilmez2017}, but after Ying et al. showed that the $\ell_1$-penalty is suboptimal for Laplacian constraints, subsequent developments were made in~\cite{zhao_2019},~\cite{ying_2020},~\cite{kumar},~\cite{ying_2021},~\cite{medvedovsky2024}. Unlike implementations such as BigQuic~\cite{bigquic} for solving the graphical lasso problem, and~\cite{kalofolias} for the more general smoothness assumption, methods that enforce Laplacian constraints with the MLE formulation scale poorly with the size of the problem. Additionally, all methods above typically use iterative optimization and require a grid search over one or multiple hyperparameters.

To avoid iterative optimization, we propose an estimator based on the connection between the Laplacian pseudoinverse and the effective resistance matrix. We use this to construct a new regularizer applied to the estimated effective resistance matrix. Effective resistance~\cite{ellens} has been used for graph estimation in~\cite{bennett} and~\cite{pavez}, but in a manner different from this paper. We also draw on previous work on the connection between effective resistance and Euclidean metrics~\cite{Devriendt_2022}.

We make the common assumption that the graphs considered are undirected and connected, but allow signed edge weights, as in, for example,~\cite{fontan2023}, provided the graph Laplacian is positive semidefinite (PSD) and has rank $p-1$. As the key idea underlying our regularizer, we show that elementwise powers of effective resistance matrices of such graphs can remain Euclidean distance matrices (EDMs) for exponents greater than one. Based on this result, we construct a direct estimator of the graph Laplacian from the sample covariance matrix, thereby avoiding iterative optimization. Our numerical experiments show that this yields performance comparable to smoothness-based estimators on the dataset used, while significantly reducing the core estimator's computational cost.

\section{Elementwise powers of effective resistances}

Before presenting our algorithm, we provide a motivation for our regularization approach. For this purpose, define a $p\times p$ EDM $D$ as a matrix with squared pairwise Euclidean distances $D_{ij}=\|y_i-y_j\|_2^2$ among points $y_i\in\mathbb{R}^d,i\in\{1,...,p\}$ of dimension $0\leq d\leq p-1$. A matrix being an EDM is equivalent to the matrix $-\frac{1}{2}JDJ$ being positive semidefinite, where $J$ is the centering matrix $J = I - \frac{1}{p}\mathbf{1}\mathbf{1}^\top$ and $\mathbf{1}$ is the column vector of all ones~\cite{dokmanic}. We will consider elementwise powers of such matrices, denoted by $^\circ$, such that $[D^{\circ\alpha}]_{ij}=D_{ij}^{\alpha}$. We restate a known property of EDMs in Lemma~\ref{thm:nons}, with derivations found in e.g.~\cite{schoenberg,Micchelli1986}.

\begin{lemma}\label{thm:nons}
Let $D$ be a $p\times p$ EDM of $p$ distinct points. Then, for all $0<\alpha<1$, $D^{\circ\alpha}$ is an EDM of full rank, $p$ distinct points, and $x^\top D^{\circ\alpha}x<0$ for all $x\neq 0$ such that $x^\top\mathbf{1}=0$.
\end{lemma}

Consider a graph $ \mathcal{G} $ and its Laplacian $L$.
The effective resistance matrix of $\mathcal{G}$ is an EDM~\cite{Devriendt_2022}, and the transformation between the effective resistance matrix and Laplacian is given by $R = \mathbf{1}\text{diag}(L^\dagger)^\top + \text{diag}(L^\dagger)\mathbf{1}^\top - 2L^\dagger$ and $L = (-\frac{1}{2}JRJ)^\dagger$, where $\text{diag}(\cdot)$ is a column vector with the diagonal elements of the matrix. Now, consider a graph whose effective resistance matrix to some elementwise power $\gamma$ remains an EDM. Denote by $\Omega_\gamma$ the set of all such graphs. Note that this means that all graphs we consider in this paper are contained in $\Omega_1$. Proposition~\ref{thm:gammamin} shows an interesting property of $\Omega_\gamma$.

\begin{proposition}\label{thm:gammamin}
    For any undirected connected graph $\mathcal{G}$ with positive semidefinite Laplacian of rank $p-1$, there exists a $\gamma>1$ such that $\mathcal{G}\in\Omega_a$ for all $a\in[1,\gamma]$.
\end{proposition}

\begin{proof}
    Since $L$ is a positive semidefinite Laplacian with rank $p-1$, $L^\dagger=-\frac{1}{2}JRJ$ is positive semidefinite with exactly one zero eigenvalue. From the continuity of roots of the characteristic polynomial, this means that there exists a $\gamma>1$ that keeps the remaining $p-1$ eigenvalues of $-\frac{1}{2}JR^{\circ\gamma}J$ positive. As mentioned previously, positive semidefiniteness of this matrix is equivalent to $R^{\circ\gamma}$ being an EDM, and by Lemma~\ref{thm:nons}, all $R^{\circ a}$ with $a\in[1,\gamma]$ will also be EDMs.
\end{proof}

Denote by $\gamma^\ast$ the supremum of the values of $\gamma$ for which Proposition~\ref{thm:gammamin} holds. The value of $\gamma^\ast$ depends on the topology of $\mathcal{G}$, its size, and edge weights. Fig.~\ref{fig:gstar_graphs} shows $\gamma^\ast$ computed using binary search for the graph families in Table~\ref{tab:graphs}, as the number of nodes grows. Since four graph families are random graphs, we generate 10 replicates for each $p$ and for each family. We will now use the fact that $\gamma^\ast>1$ for all considered graphs to construct a regularization for graph estimation.

\begin{table}
    \caption{Specification of graph families used in the paper}
    \label{tab:graphs}
    \centering
    \begin{tabular}{l|l}\hline
        Graph$^\ast$ & Description \\ \hline
        Path & Connects nodes in a chain \\
        Cycle & Adds an edge between the two endpoints of the path \\
        Star & One central node connected to all remaining nodes \\
        Grid & Two-dimensional lattice \\
        BA$^{(s)}$ & Barabási-Albert graph, preferential attachment to $s$ nodes \\
        Reg$^{(s)}$ & Random regular graph with degree $s$ \\
        ER$^{(s)}$ & Erd\H{o}s--R\'enyi graph$^{**}$ with edge probability $s$ \\
        WS$^{(s,t)}$ & Watts-Strogatz graph, degree $s$ and rewiring probability $t$ \\ \hline
        \multicolumn{2}{p{251pt}}{$^\ast$Weights are drawn from a uniform distribution on the interval $[0.5,2.0]$. $^{**}$The Erd\H{o}s--R\'enyi graph is resampled until the graph is connected.}
    \end{tabular}
\end{table}

\begin{figure}
    \centering
    \includegraphics[width=.99\linewidth]{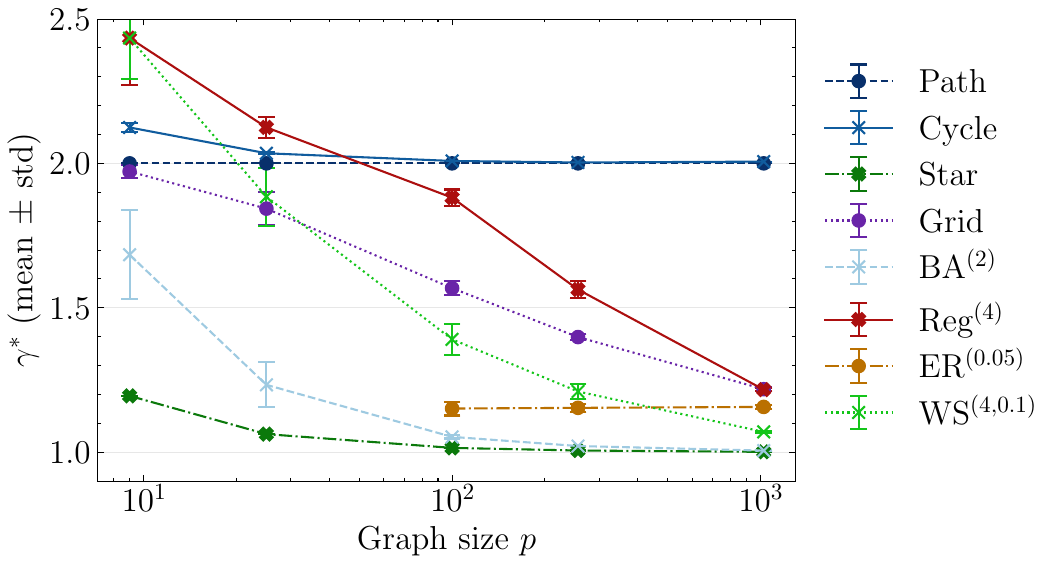}
    \caption{The value of the elementwise power~$\gamma^\ast$ for which the effective resistance matrix remains an EDM. The value of~$\gamma^\ast$ depends on the graph family, size, and edge weights.}
    \label{fig:gstar_graphs}
\end{figure}

\section{Graph estimation algorithm}

Motivated by the previous section, we propose Algorithm~\ref{alg:dsle}, which we refer to as Direct Spectral Laplacian Estimation~(DSLE). Denote by $X$ the $n\times p$ matrix of samples from $\mathcal{N}(0,L^\dagger)$ (assumed to be centered), and $S$ the corresponding sample covariance matrix. The main idea of Algorithm~\ref{alg:dsle} is to project the estimated effective resistance matrix $\hat{R}$ to the elementwise power of $\gamma$, to an EDM using Algorithm~\ref{alg:mds}, and then transform it back into an estimated Laplacian. Here, MDS$_+$ denotes multidimensional scaling in which only non-negative eigenvalues are retained. This standard method projects the centered Gram matrix of any symmetric matrix $Y$ with a zero diagonal onto the positive semidefinite cone, yielding a valid EDM~\cite{dokmanic}. The method is specified in Algorithm~\ref{alg:mds}, where $Q$ and $\Lambda$ contain the eigenvectors and eigenvalues, respectively, from the eigenvalue decomposition (EVD), and $\Lambda_+$ is equal to $\Lambda$, but with the negative eigenvalues set to zero.

\begin{algorithm}\caption{DSLE}\label{alg:dsle}
    \DontPrintSemicolon
    \KwIn{$X\in\mathbb{R}^{n\times p}$, $\gamma$}
    \KwOut{$\bar{L}$}
    \Require{$\gamma>1,\text{ }\forall\, i\neq j:\; X_{:,i}\neq X_{:,j}$}
    $S\gets\frac{1}{n}\sum_{i=1}^n X_{i,:}^\top X_{i,:}$\;
    $\hat{R}\gets\mathbf{1}\text{diag}(S)^\top + \text{diag}(S)\mathbf{1}^\top - 2S$\;
    $\bar{R}\gets \text{MDS}_+(\hat{R}^{\circ\gamma})^{\circ\frac{1}{\gamma}}$\;
    $\bar{L}\gets (-\frac{1}{2}J\bar{R}J+\frac{1}{p}\mathbf{1}\mathbf{1}^\top)^{-1}-\frac{1}{p}\mathbf{1}\mathbf{1}^\top$\;
\end{algorithm}

\begin{algorithm}
    \DontPrintSemicolon
    \KwIn{$Y$}
    \KwOut{$D$}
    $B\gets-\frac{1}{2}JYJ$\;
    $\Lambda,Q\gets \text{EVD}(B)$\;
    $G\gets Q\Lambda_{+} Q^\top$\;
    $D\gets \mathbf{1}\text{diag}(G)^\top + \text{diag}(G)\mathbf{1}^\top - 2G$\;
    \caption{MDS$_+$ (See Algorithm 1 in~\cite{dokmanic})}
    \label{alg:mds}
\end{algorithm}

To provide guarantees on the output of Algorithm~\ref{alg:dsle}, we first state Lemma~\ref{thm:increasing}, showing that Algorithm~\ref{alg:mds} is a non-decreasing operation on all elements of its input matrix $Y$. The proof is given in the supplementary material.

\begin{lemma}\label{thm:increasing}
    Let $Y$ be a symmetric matrix with zero diagonal. Then, for the output matrix $D$ in Algorithm~\ref{alg:mds}, it holds that $D_{ij}\geq Y_{ij}$ for all $i,j\in\{1,...,p\}$.
\end{lemma}

Now, Proposition~\ref{thm:main} shows that the output matrix $\bar{L}$ of Algorithm~\ref{alg:dsle} is always a signed PSD Laplacian with rank $p-1$.

\begin{proposition}\label{thm:main}
    Assume that no two columns in $X$ are identical, and $\gamma>1$. Then, $\bar{L}$ in Algorithm~\ref{alg:dsle} is a signed PSD Laplacian with rank $p-1$. Thus, the final step of Algorithm~\ref{alg:dsle} is equivalent to taking the pseudoinverse.
\end{proposition}

\begin{proof}[Proof sketch]
(Full proof is given in the supplementary material) First, if the columns of $X$ are distinct, then $\hat{R}$ has no zero off-diagonal entries, as follows from the explicit expression for each element in $\hat{R}$. By Lemma~\ref{thm:increasing}, $\mathrm{MDS}_+(\hat{R}^{\circ\gamma})$ represents $p$ distinct points. Since $0<1/\gamma<1$, the properties in Lemma~\ref{thm:nons} apply to $\bar{R}$. These conditions give $\text{rank}(-\frac{1}{2}J\bar{R}J)=p-1$, which ensures that the inverse construction in Algorithm~\ref{alg:dsle} produces a signed PSD Laplacian with rank $p-1$.
\end{proof}

We suggest that Algorithm~\ref{alg:dsle} removes noise that is inconsistent with effective resistances under the assumption of the chosen $\gamma$, and the higher the value of $\gamma$, the more restrictive this regularization is. We also note that the algorithm restores the rank deficiency when $n<p$.

To illustrate this, we generate a Reg$^{(3)}$ graph (see Table~\ref{tab:graphs}) of size $p=512$. The resulting graph has $\gamma^\ast\approx1.4$. We then take $n=256$ samples from $\mathcal{N}(0,L^\dagger)$, run Algorithm~\ref{alg:dsle} with $\gamma=1.4$, and plot the resulting edge weight estimates in a histogram (Fig.~\ref{fig:dists}(b)). For comparison, we also plot the result without regularization and replacing the last step in Algorithm~\ref{alg:dsle} with the standard pseudoinverse, in Fig.~\ref{fig:dists}(a), the MLE~\cite{medvedovsky2024} result in Fig.~\ref{fig:dists}(d), and raw correlation in Fig.~\ref{fig:dists}(c). The correlation matrix is calculated as $C_{ij}=S_{ij}/\sqrt{S_{ii}S_{jj}}$. This shows that DSLE with $\gamma=1.4$ separates most true from false edges.

\begin{figure}
    \centering
    \includegraphics[width=.99\linewidth]{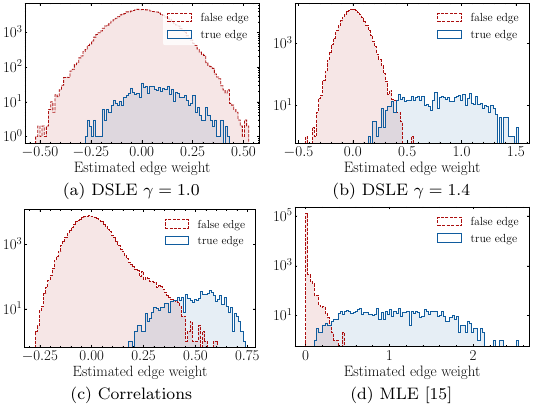}
    \caption{Weight distributions from DSLE with (a) $\gamma=1.0$, (b) $\gamma=1.4$, (c) values in the correlation matrix, and (d)~MLE~\cite{medvedovsky2024}, colored by ground truth adjacency. DSLE with $\gamma=1.4$ separates most true edges from false edges.}
    \label{fig:dists}
\end{figure}

We note that in Algorithm~\ref{alg:dsle}, matrices are dense, leading to a complexity of $\mathcal{O}(p^3)$. This is similar to other Laplacian-constrained methods, but worse than, for example, the $\mathcal{O}(p^2)$ complexity in~\cite{kalofolias}. However, our method does not incur this complexity at each step of an iterative optimization algorithm.

\begin{remark}
Bias is introduced both by the nonlinear operation on $\hat{R}$ and when $\gamma>\gamma^\ast$. Analysis of the impact of this bias is, however, beyond the scope of this paper.
\end{remark}

\section{Experimental results}

To evaluate Algorithm~\ref{alg:dsle}\footnote{Code is available at \url{https://github.com/ckjellson/GraphLapEstimation}.}, whose output is dense, we sparsify the resulting graph~$\bar{L}$ using Algorithm~\ref{alg:spars}. This is a sweep over~$M$ possible thresholds, but to guarantee connectedness, we also precompute the \textit{maximum} spanning tree (MaxST) with time complexity~$\mathcal{O}(p^2)$ using Prim's algorithm. We then let the resulting graphs be the union of the MaxST and each thresholded graph, where~$\odot$ denotes the elementwise product. Since we consider only non-negative thresholds, the PSD and the rank-$p-1$ properties are preserved, provided all MaxST weights are positive. The final step is a global rescaling that minimizes the negative log-likelihood, which follows directly from the MLE formulation in~\cite{medvedovsky2024}, see the supplementary material for a derivation. Note that without the regularization step in Algorithm~\ref{alg:dsle}, the sparsification step would not be helpful, as indicated by Fig.~\ref{fig:dists}(a).

\begin{algorithm}\caption{Sparsify (Sp)}\label{alg:spars}
    \DontPrintSemicolon
    \KwIn{$\bar{L},S,M$}
    \KwOut{$\bar{L}_1,\dots,\bar{L}_M$}
    $\bar W_{ij} \gets -\bar L_{ij}$ for $i \ne j$, and $\bar W_{ii} \gets 0$\;
    \For{$m \gets 1$ \KwTo $M$}  {
        $\tau\gets \max(\bar{W})m/M$\;
        $Q_{ij}\gets1\text{ if }(i,j)\in\text{MaxST}(\bar{W})\cup\{\bar{W}_{ij}>\tau\}\text{ else } 0$\;
        $\bar{W}_m\gets\bar{W}\odot Q$\;
        $(\bar{L}_m)_{ii} \gets \sum_j (\bar{W}_m)_{ij}$,  $(\bar{L}_m)_{ij} \gets -(\bar{W}_m)_{ij}$, $i \ne j$\;
        $\bar{L}_m\gets((p-1)/\text{trace}(\bar{L}_mS))\bar{L}_m$\;
    }
\end{algorithm}

We generate a dataset of graphs with $p=256$ for the last four graph families in Table~\ref{tab:graphs}. This selection of graphs covers both a wide range of topologies and a wide range of $\gamma^\ast$. For this purpose, we construct 10 replicates in each graph family. Fig.~\ref{fig:gammastar} shows the distribution of $\gamma^\ast$ for the graphs in the dataset determined by binary search.

\begin{figure}
    \centering
    \includegraphics[width=.9\linewidth]{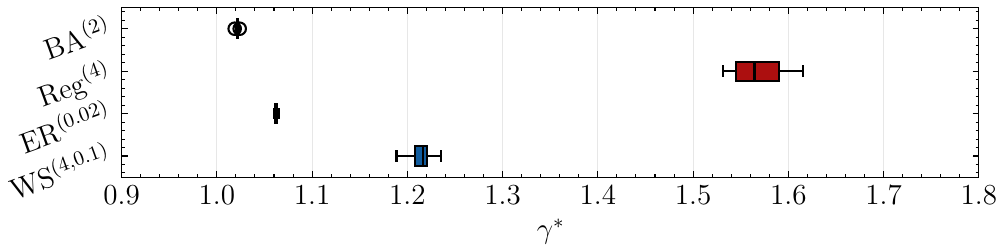}  
    \caption{Distributions of $\gamma^\ast$ for the graphs in the dataset.}
    \label{fig:gammastar}
\end{figure}

For comparison, we consider Kalofolias~\cite{kalofolias}, MLE~\cite{medvedovsky2024} and NewGLE~\cite{medvedovsky2024}. Additionally, we compare with applying Algorithm~\ref{alg:spars} to MLE~\cite{medvedovsky2024} and to the correlation matrix, which we refer to as MLE~\cite{medvedovsky2024}+Sp and CORR+Sp, respectively. To recover the scale for Kalofolias~\cite{kalofolias}, its estimates are rescaled in the same way as in the final step of Algorithm~\ref{alg:spars}. For each graph instance, we generate a \textit{path} (a sequence of estimated Laplacians with varying levels of sparsity) using the considered algorithms, except for MLE~\cite{medvedovsky2024}, which produces only a single estimate. For Kalofolias~\cite{kalofolias} and NewGLE~\cite{medvedovsky2024}, we generate paths by sweeping over 30 values of their sparsity hyperparameter, while MLE~\cite{medvedovsky2024}+Sp, CORR+Sp, and DSLE+Sp use Algorithm~\ref{alg:spars} with $M=30$, indicated by ``+Sp''.

We consider two metrics, F-score (FS) and relative error (RE). The F-score is given by $\text{FS} = \frac{2\text{tp}}{2\text{tp}+\text{fp}+\text{fn}}$, where tp=true positives, fp=false positives, and fn=false negatives. The relative error is given by $\text{RE} = ||\bar{L}_m-L||_{F}/||L||_{F}$. We evaluate only the solution along each path that minimizes $\sqrt{\smash[b]{\text{RE}^2+(1-\text{FS})^2}}$, i.e., assuming that model selection (see e.g.~\cite{model_selection}) is optimal for all algorithms. This is common when evaluating comparable methods, but for data without known ground truth, implementing model selection is necessary. Fig.~\ref{fig:results} shows the mean and standard deviations of the two metrics for each graph type and algorithm.

\begin{figure}
    \centering
    \includegraphics[width=.99\linewidth]{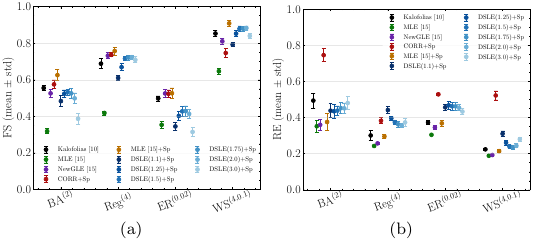}
    \caption{(a) F-score (FS) and (b) relative error (RE) for the best achievable solutions of each algorithm for each considered graph family. Despite the graphs’ low but varied values of $\gamma^\ast$, choosing $\gamma=1.5$ in DSLE yields consistently high performance.}
    \label{fig:results}
\end{figure}

We first note that the DSLE performance is comparable for $\gamma\in\{1.25, 1.5, 1.75, 2.0\}$ despite the low values of $\gamma^\ast$ reported in Fig.~\ref{fig:gammastar}. By Proposition~\ref{thm:gammamin}, we can choose $\gamma\leq\gamma^\ast$ without introducing additional bias, but this does not imply that $\gamma=\gamma^\ast$ is optimal for graph estimation performance. We conjecture that for many sparse graphs, there are small subsets of the graph that significantly lower $\gamma^\ast$, but have mostly local effects on graph estimation performance when using Algorithm~\ref{alg:dsle}.

To evaluate the dependence on uncertainty of the sample covariance matrix, we perform the above experiment for different numbers of samples according to $n/p\in\{0.25,0.5,0.75,1.0\}$. We use $\gamma=1.5$ as a heuristic and plot the best achievable FS and RE against $n/p$ across all graphs in the dataset in Fig.~\ref{fig:results2}. As expected, performance improves as $n/p$ increases. A further sensitivity analysis of $\gamma$ is provided in the supplementary material, also indicating that for our examples over the considered values of $n/p$, the sensitivity to the choice of~$\gamma$ is low. This shows that, in this setting, when implementing model selection for data without a known ground truth, we need to consider only a small number of values of~$\gamma$.

\begin{figure}
    \centering
    \includegraphics[width=.99\linewidth]{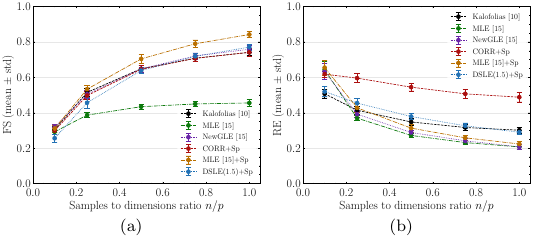}
    \caption{(a) FS and (b) RE for the best solution along each path depending on $n/p$. DSLE is simultaneously comparable to Kalofolias~\cite{kalofolias} in both edge and weight recovery.}
    \label{fig:results2}
\end{figure}

When comparing to CORR+Sp in Fig.~\ref{fig:results2}, we can see that DSLE has not only recovered correlation between edges but also their relative strength, since RE is substantially lower than for CORR+Sp. Another interesting observation is that MLE~\cite{medvedovsky2024}+Sp performs at least on par with the results obtained by sweeping over the sparsity hyperparameter as in NewGLE~\cite{medvedovsky2024}. This may be because the minimax concave penalty is itself threshold-based.

To investigate the computational cost of Algorithm~\ref{alg:dsle}, we implement DSLE in MATLAB and use publicly available MATLAB code for Kalofolias~\cite{kalofolias} and MLE~\cite{medvedovsky2024} on a laptop with an Intel Core Ultra 7 255U CPU and 16 GB RAM. MLE~\cite{medvedovsky2024} is, to our knowledge, among the fastest implementations of Laplacian-constrained MLE. Furthermore, we find the runtime of NewGLE~\cite{medvedovsky2024} to be independent of the chosen sparsity parameter, and MLE~\cite{medvedovsky2024} uses the same implementation. However, the Kalofolias~\cite{kalofolias} runtime depends on its sparsity parameter, so we report only the time to compute the most accurate Laplacian. Table~\ref{tab:compl} shows the resulting runtimes for four different graph sizes over ten replicates of a Barabási-Albert graph with the same parameters as used in the previous section. This shows that the runtime of DSLE is significantly lower than that of the other core estimators for the considered graphs. We note that there are many ways of implementing grid search over sparsity parameters in Kalofolias~\cite{kalofolias} and NewGLE~\cite{medvedovsky2024}, as well as choosing $\gamma$ in DSLE, and subsequent sparsification and/or model selection. Therefore, this paper compares only the runtimes of the core estimators.

\begin{table}
    \caption{Single-run core-estimator runtimes for Kalofolias~\cite{kalofolias}, MLE~\cite{medvedovsky2024}, and DSLE for 10 instances of a Barabási-Albert graph of different sizes $p$. DSLE is at least an order of magnitude faster for the considered graphs.}
    \label{tab:compl}
    \centering
    \begin{tabular}{l|l|l|l}\hline
        $p$       & Kalofolias~\cite{kalofolias} & MLE~\cite{medvedovsky2024} & DSLE \\ \hline
        $256$ & 0.26$\pm$0.01s & 0.55$\pm$0.06s & \textbf{0.026}$\pm$0.01s \\
        $512$ & 2.3$\pm$0.1s & 5.9$\pm$1.3s & \textbf{0.06}$\pm$0.01s \\
        $1024$ & 8.8$\pm$0.3s & 47$\pm$10s & \textbf{0.30}$\pm$0.04s \\
        $2048$ & 42$\pm$0.6s & 780$\pm$160s & \textbf{2.3}$\pm$0.05s \\ \hline
    \end{tabular}
\end{table}

We conclude this section by suggesting that our dense estimator, together with the sparsification approach, yields performance comparable to that of the smoothness-based estimator on our dataset, as shown in Fig.~\ref{fig:results2}. In addition, to our knowledge, it is the only algorithm that can achieve this for both FS and RE simultaneously without iterative optimization.

\section{Conclusion}

In this paper, we have investigated the effect of elementwise powers of the effective resistance matrix of a graph, quantified by the graph property $\gamma^\ast$. Using this theory, we have proposed an estimator for graph Laplacian estimation and evaluated it using a simple sparsification approach. The main features of this approach are its simplicity, which leads to low computational cost in the small-to-moderate regime, and its independence of initialization and convergence criteria. We therefore suggest using the method for problems on graphs with up to a few thousand nodes when computational resources are limited, or as an initialization for other approaches.

Future work includes statistical analysis of our approach and evaluation on real data under model mismatch. It will also be meaningful to evaluate the method on ground-truth graphs with negative weights and further characterize $\gamma^\ast$.

\newpage
\IEEEtriggeratref{14}
\bibliographystyle{IEEEtran}
\bibliography{bib}

\newpage

\pagestyle{empty}

\section{Supplementary material}

\subsection{Proof of Lemma~\ref{thm:increasing}}

Consider an eigendecomposition of the matrix $B$ in Algorithm~\ref{alg:mds}, resulting in
\begin{equation*}
\begin{split}
    B_{ij} & = \sum_{k}\lambda_kv_{ik}v_{jk}\\
    & =\sum_{k:\lambda_k\geq0}\lambda_kv_{ik}v_{jk} + \sum_{k:\lambda_k<0}\lambda_kv_{ik}v_{jk}\\
    & =G_{ij} + \sum_{k:\lambda_k<0}\lambda_kv_{ik}v_{jk},
\end{split}
\end{equation*}
where $\lambda_k$ is the k'th eigenvalue and $v_{ik}$ is the i-th element of the k-th eigenvector. The last equality holds because $G_{ij}$, by construction, only contains the nonnegative eigenvalues. Substituting this into the corresponding formulas for the distance matrices, we get
\begin{equation*}
    D_{ij} = G_{ii} + G_{jj} - 2G_{ij},
\end{equation*}
and
\begin{equation*}
\begin{split}
    Y_{ij} & = B_{ii} + B_{jj} - 2B_{ij}\\
    & = G_{ii} + G_{jj} - 2G_{ij} + \sum_{k:\lambda_k<0}\lambda_k(v_{ik}^2+v_{jk}^2 - 2v_{ik}v_{jk})\\
    & =D_{ij} + \sum_{k:\lambda_k<0}\lambda_k(v_{ik}-v_{jk})^2,
\end{split}
\end{equation*}
where the first equality holds only if $Y$ is symmetric and has exclusively zeros on the diagonal. Since the last term is less than or equal to zero, it holds that $D_{ij}-Y_{ij}\geq0$ for all $i,j$.\qed

\subsection{Proof of Proposition~\ref{thm:main}}

We start by showing that $\hat{R}$ in Algorithm~\ref{alg:dsle} contains no off-diagonal zeros. Since columns in $X$ are assumed to be centered, each element in $\hat{R}$ is given by
\begin{equation*}
\begin{aligned}
    \hat R_{ij}
    & = S_{ii}+S_{jj}-2S_{ij} \\
    & = \frac{1}{n}X_{:,i}^\top X_{:,i} + \frac{1}{n}X_{:,j}^\top X_{:,j} - \frac{2}{n}X_{:,i}^\top X_{:,j} \\
    & = \frac{1}{n}\left(\|X_{:,i}\|_2^2 + \|X_{:,j}\|_2^2 - 2X_{:,i}^\top X_{:,j}\right) \\
    & = \frac{1}{n}\|X_{:,i}-X_{:,j}\|_2^2.
\end{aligned}
\end{equation*}
Thus, if there are no two identical columns in $X$, there are no off-diagonal zeros in $\hat{R}$, and therefore also in $\hat{R}^{\circ\gamma}$.

Next, $\text{MDS}_+(\hat{R}^{\circ\gamma})$ is an EDM by construction from Algorithm~\ref{alg:mds}. By Lemma~\ref{thm:increasing}, no element in $\text{MDS}_+(Y)$ can be smaller than the corresponding element in $Y$. Hence, there are also no off-diagonal zeros in $\text{MDS}_+(\hat{R}^{\circ\gamma})$, and $\text{MDS}_+(\hat{R}^{\circ\gamma})$ is thus made up of $p$ distinct points. Since $\gamma>1$ we have $0<1/\gamma<1$, and thus we can apply Lemma~\ref{thm:nons}.

We want to show that $J\bar{R}J$, with $\bar{R}$ in Algorithm~\ref{alg:dsle}, has rank $p-1$. By Lemma~\ref{thm:nons} it holds that $x^\top\bar{R}x<0$ for all $x\neq 0$ such that $x^\top\mathbf{1}=0$. Also, $Jx=x$ for all such $x$, which gives
\begin{equation*}
    x^\top(J\bar{R}J)x=x^\top\bar{R}x<0,\text{ }\forall x\neq 0,x^\top\mathbf{1}=0.
\end{equation*}
This means that the quadratic form for $J\bar{R}J$ and $\bar{R}$ is the same on $\mathbf{1}^\bot$, and from before we know that $\bar{R}$ is nonsingular, so $\text{rank}(J\bar{R}J)=p-1$. The matrix $-\frac{1}{2}J\bar{R}J$ is thereby symmetric PSD with null-space $\text{span}\{\mathbf{1}\}$, which means that, as for example in~\cite{fontan2023}, we can compute the pseudoinverse by
\begin{equation*}
    (-\frac{1}{2}J\bar{R}J)^\dagger = (-\frac{1}{2}J\bar{R}J+\frac{1}{p}\mathbf{1}\mathbf{1}^\top)^{-1}-\frac{1}{p}\mathbf{1}\mathbf{1}^\top.
\end{equation*}
Thus, $\bar{L}$ is positive semidefinite, has rank $p-1$ and, by construction using double centering, $\bar{L}\mathbf{1}=0$. The matrix $\bar{L}$ is thereby a signed Laplacian of a connected graph.\qed

\subsection{Derivation of the scaling factor used in Algorithm~\ref{alg:spars}}

Denote by $\bar{L}$ an estimated Laplacian. Then, the sought scaling factor $\alpha>0$ that minimizes the negative log-likelihood is given by
\begin{equation*}
    \argmin_\alpha f(\alpha),\quad f(\alpha) = \text{Tr}(\alpha\bar{L}S) - \text{log}|\alpha\bar{L}|_+,
\end{equation*}
where $\text{Tr}(\cdot)$ and $|\cdot|_+$ denote the trace and pseudo-determinant, respectively, just as in the MLE formulation in~\cite{medvedovsky2024}. Since the rank of $\bar{L}$ is $p-1$, we can rewrite and minimize $f(\alpha)$ according to
\begin{equation*}
    \begin{aligned}
        & f(\alpha) = \alpha \text{Tr}(\bar{L}S) - (p-1)\text{log}(\alpha) - \text{log}|\bar{L}|_+\\
        \implies & f'(\alpha)=0\quad\iff\quad\text{Tr}(\bar{L}S)-\frac{p-1}{\alpha}=0,
    \end{aligned}
\end{equation*}
resulting in the scaling factor used in Algorithm~\ref{alg:spars}.

\subsection{Extended sensitivity analysis for the choice of $\gamma$}

To further evaluate the sensitivity to the parameter~$\gamma$ in Algorithm~\ref{alg:dsle}, we provide $n/p$-plots for FS and RE for $\gamma\in\{1.1, 1.25, 1.5, 1.75, 2.0, 3.0\}$ in Fig.~\ref{fig:results3}. These were generated in the same fashion as the plots in Fig.~\ref{fig:results2}, but now only comparing DSLE for different choices of $\gamma$. This shows that, in our experimental setup, performance degrades significantly only in the extreme cases $\gamma=1.1$ and $\gamma=3.0$, regardless of the value of $n/p$. Also, the performance was already shown to be stable over the different graph topologies in Fig.~\ref{fig:results} for the same values of $\gamma$, i.e., except for $\gamma=1.1$ and $\gamma=3.0$.

\begin{figure}[H]
    \centering
    \includegraphics[width=.99\linewidth]{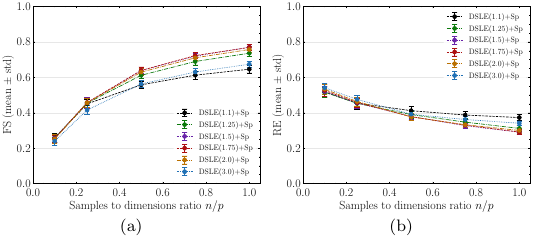}
    \caption{(a) FS and (b) RE for the best solution along each path depending on $n/p$ for varying $\gamma$.}
    \label{fig:results3}
\end{figure}

\end{document}